\documentclass[twocolumn]{article}
\usepackage{geometry}

\input{packages.sty}
\input{commands.sty}

\begin{document}

\author{Holger Heidrich$^1$, Jannik Irmai$^1$, Bjoern Andres$^{1,2}$}
\title{\bf A 4-approximation algorithm for min max correlation clustering}
\date{$^1$\textit{TU Dresden} \hspace{2ex} $^2$\textit{Center for Scalable Data Analytics and AI Dresden/Leipzig}}

\maketitle

\begin{abstract}
We introduce a lower bounding technique for the min max correlation clustering problem and, based on this technique, a combinatorial $4$-approximation algorithm for complete graphs.
This improves upon the previous best known approximation guarantees of $5$, using a linear program formulation \citep{kalhan2019correlation}, and $40$, for a combinatorial algorithm \citep{davies2023fast}.
We extend this algorithm by a greedy joining heuristic and show empirically that it improves the state of the art in solution quality and runtime on several benchmark datasets.
\end{abstract}

\section{INTRODUCTION}

Correlation clustering refers to the task of clustering elements based on pairwise similarity.
The objective is to find a partition of the set of elements such that, preferably, similar elements are in the same cluster and dissimilar elements are in distinct clusters.
The number and size of clusters is not predefined but determined by the pairwise similarities.
Originally, the correlation clustering problem is defined by \citet{bansal2004correlation} for a graph where each edge $\{u,v\}$ is labeled either $+$ or $-$, depending on whether $u$ and $v$ are deemed similar or dissimilar.
For any partition of the node set, an edge $\{u,v\}$ is said to be in disagreement with the partition if the edge is labeled $-$ and the nodes $u$ and $v$ are in the same cluster of the partition, or if the edge is labeled $+$ and the nodes $u$ and $v$ are in distinct clusters of the partition.
\Citet{bansal2004correlation} study the problem of finding a partition that minimizes the number of disagreeing edges. 
More recently, \citet{puleo2016correlation} have introduced the correlation clustering problem with locally bounded errors.
They define for each node in the graph the disagreement of that node as the number of edges that are incident to that node and in disagreement with the partition.
They study the problem of finding a partition that minimizes a function of the disagreements of the nodes.
Among others, they consider $\ell^p$ norms for $p \geq 1$.
In the special case of $p=\infty$, this leads to the problem of finding a partition such that the maximum disagreement over all nodes is minimal.
This special case, called \emph{min max correlation clustering}, has attracted attention as a model of a sense of fairness \citep{davies2023fast}.
Here, we concentrate on the special case of complete graphs, i.e., every pair of nodes is either a $+$ or a $-$ edge.

As a first contribution, we present a lower bounding technique for the min max correlation clustering problem for complete graphs. 
This lower bound can be computed efficiently, by a combinatorial algorithm, and does not require solving a linear program (LP).
Moreover, this bound is different from that obtained by solving the canonical LP relaxation, i.e.~instances exist where either of these bounds is strictly stronger.
As a second contribution, we use this lower bounding technique to derive a $4$-approximation algorithm for the min max correlation clustering problem for complete graphs.
This improves upon the previous best known approximation guarantees of $5$, using an LP formulation \citep{kalhan2019correlation}, and $40$, for a combinatorial algorithm \citep{davies2023fast}.
We briefly discuss generalizations to non-complete and weighted graphs. 
However, we do not establish any approximation guarantees for these cases.
As a third contribution, we extend the $4$-approximation by a local search algorithm that is designed in due consideration of the insights provided by the lower bound.
Empirically, we show: The lower bound and the $4$-approximation with the local search extension outperform the current state of the art in both solution quality and runtime on a variety of benchmark datasets.
In particular, we present lower bounds and approximate solutions for large graphs out of reach of previous methods.

The remainder of the article is organized as follows.
In \Cref{sec:related-work}, we discuss related work. 
In \Cref{sec:problem}, we state the min max correlation clustering problem formally.
In \Cref{sec:bound}, we present the lower bounding technique and derive a $4$-approximation algorithm.
In \Cref{sec:algorithm}, we extend the $4$-approximation by an efficient local search heuristic.
In \Cref{sec:experiments}, we examine the bound and the approximation algorithm empirically, on several benchmark datasets, and compare these to the state of the art.
In \Cref{sec:conclusion}, we draw conclusions and discuss perspectives for future work.
Implementations of all our algorithms and the complete code for reproducing the experiments are available as supplementary material and at \texttt{\url{https://github.com/JannikIrmai/min-max-correlation-clustering}}.
\section{RELATED WORK}\label{sec:related-work}

The correlation clustering problem is introduced originally by \citet{bansal2004correlation}.
It is closely related to the \emph{clique partitioning} problem \citep{groetschel1989cutting,groetschel1990facets} and the \emph{graph partition} problem \citep{chopra1993partition} which ask for partitions of a graph with positive and negative edge weights that minimize the costs of the edges within clusters and between clusters, respectively.
The graph partition problem is also known as the \emph{multicut} problem \citep{deza1992clique},  not to be confused with the \emph{multiterminal cut} problem \citep{dahlhaus1992complexity} or \emph{multicommodity cut} problem \citep{leighton1999multicommodity}.
The complexity and hardness of approximation of the correlation clustering problem are studied, among others, by \citet{bansal2004correlation,demaine2006correlation,ailon2008aggregating,voice2012coalition,bachrach2013optimal,veldt2022correlation,klein2023correlation}.
The best known approximation guarantee for the correlation clustering problem in unweighted complete graphs is $1.996+\epsilon$, due to \citet{cohen2022correlation}.

\citet{puleo2016correlation} propose a generalization of the correlation clustering problem in which the objective is a function of the disagreements of the nodes.
This is motivated by the idea of bounding disagreements locally in order to model a sense of fairness and, e.g., penalize partitions in which individual nodes have a disproportionally large disagreement.
As a special case, they introduce the min max correlation clustering problem, show that it is $\textsc{np}$-hard for complete graphs, and provide a $48$-approximation algorithm.
In a subsequent study, \citet{charikar2017local} present approximation algorithms for several variants of the correlation clustering problem, including a $7$-approximation for the min max objective.
This is further improved by \citet{kalhan2019correlation} who, among other results, present a $5$-approximation algorithm for the correlation clustering problem for complete graphs with the $\ell^p$ objective for all $p \geq 1$, which includes the min max objective ($p = \infty$).
Their algorithm is based on rounding the solution of the canonical LP relaxation of the correlation clustering problem.
As solving an LP can be computationally expensive, \citet{davies2023fast} propose a method for computing a feasible solution to this specific LP more efficiently.
While this solution is not guaranteed to be optimal for the LP, its objective is at most 8 times greater than that of the optimal integral solution for the min max objective. 
This results in a $40$-approximation for the min max correlation clustering problem for complete graphs, which is the first approximation that does not require solving an LP.
They show empirically that their algorithm performs well in practical applications and scales to instances previously out of reach.
The arguments we use in order to prove the existence and efficiency of a combinatorial 4-approximation algorithm (\Cref{thm:four-approximation}) are closely related to the non-local charging argument of \citet{davies2023fast}. Our arguments are combinatorial while \citet{davies2023fast} use a correlation metric (i.e. a fractional solution to the LP relaxation) to prove their result. 
To the best of our knowledge, nothing is known about the hardness of approximation for the min max correlation clustering problem.

The more general $\ell_p$ objective is further studied by \citet{jafarov2020correlation} who present approximation algorithms for complete weighted graphs with a bounded weight range. \citet{davies2023one} establish an algorithm for computing a partition that simultaneously approximates all $\ell_p$-norm objectives within a constant factor.

Another closely related variant of the correlation problem is studied by \citet{ahmadi2019min}.
They consider the objective of minimizing the maximum over the disagreements within all clusters and develop a $\mathcal{O}(\log(n))$-approximation algorithm.
An improved approximation guarantee of $2+\epsilon$ is given by \citet{kalhan2019correlation} who also establish inapproximability within a factor better than $2$ assuming the unique games conjecture (UGC).
Hence, their approximation guarantee is the best possible if UGC holds.

\section{MIN MAX CORRELATION CLUSTERING FOR COMPLETE GRAPHS}\label{sec:problem}

To begin with, we state the min max correlation clustering problem for complete graphs formally, using elementary notation.
For any set $V$, let $P_V$ denote the set of all partitions of $V$.
For any $\Pi \in P_V$ and any $u \in V$, let $[u]_\Pi$ denote the unique $U \in \Pi$ such that $u \in U$.
For any graph $G = (V, E)$ and any $v \in V$, let us refer to
$N_v := \{v\} \cup \{w \in V \mid \{v,w\} \in E\}$
as the neighborhood of $v$ in $G$, including $v$ itself.

\begin{definition}
\label{definition:mmcc}
For any graph $G = (V, E)$, the instance of the \emph{min max correlation clustering problem} with respect to $G$ has the form
\begin{align}\label{eq:mmcc}
    \min_{\Pi \in P_V}
        \underbrace{
            \max_{v \in V}
                \quad 
                    \left| [v]_\Pi \triangle N_v \right|
        }_{=:\ \varphi(\Pi)}
    \tag{MMCC}
\end{align}
where $A \triangle B = (A \setminus B) \cup (B \setminus A)$ is the symmetric difference of sets.
Here, $\left| [v]_\Pi \triangle N_v \right|$ is the disagreement of node $v$ with partition $\Pi$ and $\varphi(\Pi)$ is the maximum disagreement of $\Pi$.
\end{definition}

\ref{eq:mmcc} can be understood as a problem with respect to an edge signed complete graph.
More specifically, one can identify the graph $(V,E)$ in \Cref{definition:mmcc} with the edge signed complete graph $(V, E^+ \cup E^-)$ such that $E^+ = E$.

\section{COMBINATORIAL LOWER BOUND AND 4-APPROXIMATION}\label{sec:bound}

In this section, we present a lower bounding technique for the min max correlation clustering problem for complete graphs.
From this technique, we derive a $4$-approximation algorithm.

We begin by establishing properties of partitions and their maximal disagreement.
The following lemma states that two nodes whose neighborhoods are particularly similar (respectively dissimilar) must be in the same cluster (respectively different clusters) in all partitions whose maximal disagreement is small enough.

\begin{lemma}\label{lem:main-lemma}
    Let $G=(V, E)$ be a graph. 
    For every partition $\Pi \in P_V$ and every $u, v \in V$:
    \begin{enumerate}[nosep,label=(\alph*)]
        \item \label{item:cluster-lower-bound}
        If  $|N_u \cap N_v| > 2 \varphi(\Pi)$ then $[u]_\Pi = [v]_\Pi$.
        \item \label{item:cluster-upper-bound}
        If $|N_u \triangle N_v| > 2 \varphi(\Pi)$ then $[u]_\Pi \neq [v]_\Pi$.
    \end{enumerate}    
\end{lemma}

\begin{proof}\let\qed\relax
    To prove \ref{item:cluster-lower-bound}, we show that $[u]_\Pi \neq [v]_\Pi$ implies $|N_u \cap N_v| \leq 2 \varphi(\Pi)$.
    By definition,
    $
        \varphi(\Pi) \geq |N_u \triangle [u]_\Pi| \geq |N_u \setminus [u]_\Pi| \geq |(N_u \cap N_v) \setminus [u]_\Pi| 
    $.
    Analogously, $\varphi(\Pi) \geq |(N_u \cap N_v) \setminus [v]_\Pi|$.
    The assumption $[u]_\Pi \neq [v]_\Pi$ and the fact that clusters of $\Pi$ are disjoint implies the desired
    $
        |N_u \cap N_v| \leq |(N_u \cap N_w) \setminus [u]_\Pi| + |(N_u \cap N_w) \setminus [v]_\Pi| \leq 2 \varphi(\Pi) 
    $.

    The second statement can be derived from Proposition 4.1 of \citet{davies2023fast}. 
    For completeness, we provide a self-contained prove below.
    We show that $[u]_\Pi = [v]_\Pi$ implies $|N_u \triangle N_v| \leq 2 \varphi(\Pi)$.
    Let $C = [u]_\Pi = [v]_\Pi$.
    By definition:
    \begin{align*}
        |N_u \triangle N_v| 
        &= |N_u \setminus N_v| + |N_v \setminus N_u| \\
        &= |(N_u \setminus N_v) \cap C| + |(N_u \setminus N_v) \setminus C| \\
        	&\quad + |(N_v \setminus N_u) \cap C| + |(N_v \setminus N_u) \setminus C| \\
        &\leq |C \setminus N_v| + |N_u \setminus C| + |C \setminus N_u| + |N_v \setminus C| \\
        &= |N_u \triangle C| + |N_v \triangle C| \\
        &= |N_u \triangle [u]_\Pi| + |N_v \triangle [v]_\Pi| \leq 2 \varphi(\Pi) \enspace . \qquad\square
    \end{align*}
\end{proof}

If there exists a partition $\Pi$ with a given maximal disagreement, \Cref{lem:main-lemma} can imply that certain nodes are in the same cluster and certain nodes are in distinct clusters.
These constraints are captured in the following definition.

\begin{definition}\label{def:lower-upper-cluster-bound}
    Let $G = (V, E)$ be a graph and let $d \in \mathbb{N}$.
    Let $G_d = (V, E_d)$ with $E_d = \{\{u,v\} \in \binom{V}{2} \mid |N_u \cap N_v| > 2d \}$ be the graph with those pairs of nodes of $G$ as edges whose neighborhoods intersect in more than $2d$ nodes.
    Let $\Pi_d \in P_V$ be the partition of $V$ into the maximal connected components of $G_d$.
    For a cluster $C \in \Pi_d$, let $U^d_C = \{v \in V \mid |N_w \triangle N_u| \leq 2d \; \forall w \in [v]_{\Pi_d}\;\forall u \in C\}$.
\end{definition}

Note that the clusters of $\Pi_d$ are precisely the sets of nodes that, by \Cref{lem:main-lemma} \ref{item:cluster-lower-bound}, must be in the same cluster of any partition $\Pi$ with maximum disagreement of $\varphi(\Pi) = d$. 
For all clusters $C \in \Pi_d$, the set $V \setminus U^d_C$ is precisely the set of nodes that, by \Cref{lem:main-lemma} \ref{item:cluster-upper-bound}, cannot be in the cluster containing $C$ of any partition $\Pi$ with $\varphi(\Pi) = d$.

\begin{theorem}\label{thm:clb}
    Let $G = (V, E)$ be a graph.
    The smallest $d \in \mathbb{N}$ with $C \subseteq U^d_C$ for all $C \in \Pi_d$ and
    \begin{align}\label{eq:disagreement-bound}
        \max_{v \in V} \;\; \Bigl|N_v \setminus U^d_{[v]_{\Pi_d}} \Bigr| + \Bigl|[v]_{\Pi_d} \setminus N_v\Bigr| \quad \leq \quad d \enspace,
    \end{align}
    is a lower bound for the min max correlation clustering problem with respect to graph $G$. 
    We call this value the \emph{combinatorial lower bound} and denote it by $\clb(G)$.
\end{theorem}

\begin{proof}
    Let $\Pi \in P_V$ be a partition with maximum disagreement $d = \varphi(\Pi)$.
    By \Cref{lem:main-lemma} and \Cref{def:lower-upper-cluster-bound}, the clusters of $\Pi$ are unions of clusters of $\Pi_d$ that are contained in the sets $U^d_C$ for $C \in \Pi_d$, i.e. $[v]_{\Pi_d} \subseteq [v]_\Pi \subseteq U^d_{[v]_{\Pi_d}}$ for all $v \in V$.
    With this, the disagreement of any node $v \in V$ with $\Pi$ can be bounded by
    \begin{align}\label{eq:node-disagreement-bound}
        |N_v \triangle [v]_\Pi| 
        &= |N_v \setminus [v]_\Pi| + |[v]_\Pi \setminus N_v| 
        \nonumber\\
        &\geq |N_v \setminus U^d_{[v]_{\Pi_d}} | + |[v]_{\Pi_d} \setminus N_v| 
        \enspace .
    \end{align}

    Now, let $d = \clb(G)$ be the smallest $d \in \mathbb{N}$ that satisfies the conditions in the statement of the theorem.
    Suppose there exists a partition $\Pi \in P_V$ with maximum disagreement $\varphi(\Pi) = d' < d$.
    By definition of $d$, there exists $v \in V$ with $|N_v \setminus U^d_{[v]_{\Pi_{d'}}} | + |[v]_{\Pi_{d'}} \setminus N_v| > d'$. 
    By \eqref{eq:node-disagreement-bound}, this implies $\varphi(\Pi) \geq |N_v \triangle [v]_\Pi| > d'$ in contradiction to the assumption.
\end{proof}

\begin{theorem}\label{thm:clb-complexity}
    Let $G = (V, E)$ be a graph.
    The combinatorial lower bound $\clb(G)$ can be computed in time $\mathcal{O}(n^2\log_2(\delta) + n\delta^2)$ where $n$ is the number of nodes in $G$, and $\delta$ is the maximum degree of all nodes in $G$.
\end{theorem}

\begin{proof}
    To begin with, we show that for all pairs of nodes the size of the intersections of their neighborhoods can be computed in time $\mathcal{O}(n^2 + n\delta^2)$.
    This can be done by the following algorithm:
    For all $u, v \in V$, let $I_{\{u,v\}} = 0$.
    For every $w \in V$ and every $\{u, v\} \in N_w$, increase $I_{\{u,v\}}$ by one to account for the fact that $w$ is in $N_u \cap N_v$.
    Then, $I_{\{u,v\}} = |N_u \cap N_v|$, and clearly, this can be done in time $\mathcal{O}(n^2 + n\delta^2)$.

    Next, we show that for a given $d \in \mathbb{N}$ it can be decided in time $\mathcal{O}(n^2)$ whether \eqref{eq:disagreement-bound} holds.
    The partition $\Pi_d$ consists of the connected components of the graph whose edges are all pairs of nodes $\{u,v\} \in \binom{V}{2}$ with $I_{\{u, v\}} > 2d$.
    These connected components can be computed, for example, by breadth first search in time $\mathcal{O}(n^2)$.
    For every node $v \in V$, the value $|[v]_{\Pi_d} \setminus N_v|$ can clearly be computed in time $\mathcal{O}(n)$.
    It remains to compute the value $|N_v \setminus U^d_{[v]_{\Pi_d}}|$ for every $v \in V$.
    To this end, we construct an auxiliary graph $G' = (V', E')$ whose nodes are the clusters in $\Pi_d$, and two clusters $C, C' \in \Pi_d$ are connected by an edge if and only if $|N_u \triangle N_{u'}| \leq 2d$ for all $u \in C$ and $u' \in C'$.
    Clearly, this graph can be computed in time $\mathcal{O}(n^2)$.
    By definition of $U^d_{[v]_{\Pi_d}}$:
    \begin{align}\label{eq:cluster-complement-linear-time}
        |N_v \setminus U^d_{[v]_{\Pi_d}}| = |\{u \in N_v \mid \{[u]_{\Pi_d}, [v]_{\Pi_d}\} \notin E'\}| \enspace,
    \end{align}
    which is the number of neighbors $u$ of $v$ whose cluster is not connected to the cluster of $v$ in $G'$.
    By representing $G'$ with an adjacency matrix, \eqref{eq:cluster-complement-linear-time} can be evaluated in time $\mathcal{O}(n)$ for each node.

    Lastly, the minimal $d$ that satisfies \eqref{eq:disagreement-bound} can be computed by a bisection algorithm.
    Clearly, $0 \leq \clb(G) \leq \delta$.
    Thus, $\mathcal{O}(\log_2(\delta))$ bisection steps are sufficient.
\end{proof}
This worst case time complexity of $\mathcal{O}(n^2\log_2(\delta) + n\delta^2) \subseteq \mathcal{O}(n^3)$ for computing the combinatorial lower bound is smaller than that of solving the LP relaxation with $\mathcal{O}(n^2)$ variables and $\mathcal{O}(n^3)$ constraints.

Below, \Cref{lem:unique-cluster} and \Cref{cor:unique-cluster} state that for every partition whose maximal disagreement is less than a quarter of the size of the neighborhood of a given node, the cluster of that node in the partition is uniquely determined.
Afterward, \Cref{thm:four-approximation} states that this yields a $4$-approximation algorithm for min max correlation clustering.

\begin{lemma}\label{lem:unique-cluster}
    Let $G = (V, E)$ be a graph and let $d = \clb(G)$ be the combinatorial lower bound according to \Cref{thm:clb}.
    For any node $v \in V$ with $|N_v| > 4 d$:
    \begin{align*}
        [v]_{\Pi_d} = U^d_{[v]_{\Pi_d}} = \{u \in V \mid |N_u \cap N_v| > |N_v|/2\} \enspace .
    \end{align*}
\end{lemma}

\begin{proof}
    By definition of the combinatorial lower bound in \Cref{thm:clb}: $[v]_{\Pi_d} \subseteq U^d_{[v]_{\Pi_d}}$.
    To prove the claim, it remains to show $U^d_{[v]_{\Pi_d}} \subseteq \{u \in V \mid |N_u \cap N_v| > |N_v|/2\} \subseteq [v]_{\Pi_d}$.

    By the assumption that $|N_v| > 4 d$, we have $|N_v|/2 > 2 d$, which implies that every $u \in V$ that satisfies $|N_u \cap N_v| > |N_v|/2$ is in $[v]_{\Pi_d}$, by definition of $\Pi_d$. 
    This yields the second inclusion.

    For $u \in V$ with $|N_u \cap N_v| \leq |N_v|/2$, we have
    \begin{align*}
        |N_u \triangle N_v| 
        &= |N_u \cup N_v| - |N_u \cap N_v| \\
        &\geq |N_v| - |N_v|/2 = |N_v|/2 > 2 d
        \enspace ,
    \end{align*}
    and \Cref{lem:main-lemma} \ref{item:cluster-upper-bound} implies $u \notin U^d_{[v]_{\Pi_d}}$.
    This yields the first inclusion.
\end{proof}

\begin{corollary}\label{cor:unique-cluster}
    Let $G = (V, E)$ be a graph. For every partition $\Pi \in P_V$ and every node $v \in V$ with $|N_v| > 4 \varphi(\Pi)$: $[v]_\Pi = \{u \in V \mid |N_u \cap N_v| > |N_v|/2\}$.
\end{corollary}

\begin{theorem}\label{thm:four-approximation}
    There exists a $4$-approximation algorithm for the min max correlation clustering problem. 
    The $4$-approximation can be computed in time $\mathcal{O}(n^2 + n \delta^2)$ where $n$ is the number of nodes and $\delta$ is the largest degree.
\end{theorem}

\begin{proof}
    The algorithm starts with the partition into singleton clusters.
    Iteratively, a node $v$ with largest disagreement with respect to the current partition is selected.
    If there exists a partition with maximum disagreement strictly less than $|N_v|/4$, then the cluster of $v$ is uniquely determined by $C = \{u \in V \mid |N_u \cap N_v| > |N_v|/2\}$, according to \Cref{cor:unique-cluster}.
    If the cluster $C$ is in conflict with any previously computed cluster (i.e.~it contains nodes that have been assigned to a different cluster in an earlier iteration) or it contains a node whose disagreement with respect to that cluster is greater than $|N_v|/4$, then there cannot exists a partition with disagreement less than $|N_v|/4$.
    This implies that the current partition is a $4$-approximation, and the algorithm terminates.
    Otherwise, $C$ is added to the current partition and the algorithm continues.

    As in \Cref{thm:clb-complexity}, the size of the intersections of neighborhoods of all pairs of nodes can be computed in $\mathcal{O}(n^2 + n \delta^2)$.
    The algorithm described above terminates after at most $\mathcal{O}(n)$ iterations.
    In each iteration, the cluster $C$ can be computed in $\mathcal{O}(n)$.
    The disagreement of a node with a given cluster can be computed in $\mathcal{O}(\delta)$.
    As each node is assigned to a cluster at most once, the disagreement needs to be computed for at most $\mathcal{O}(n)$ nodes. 
    Together, this implies the claimed runtime.
\end{proof}

\subsection{Relation to LP bound}

The combinatorial lower bound from \Cref{thm:clb} differs from the bound obtained by solving the canonical LP relaxation of the min max correlation clustering problem.
In this section, we show by \Cref{example:combinatorial-bound-better,example:lp-bound-better} that neither bound is stronger.
For an in-depth discussion of the LP bound, we refer to \citet{kalhan2019correlation}.
An efficient algorithm for computing the combinatorial lower bound is detailed in the proof of \Cref{thm:clb-complexity}.

\begin{figure}
    \centering
    \begin{subfigure}{0.49\columnwidth}
        \centering
        \begin{tikzpicture}
    \node[vertex] (0) at (0, 0) {\scriptsize 0};
    \node[vertex] (1) at (1, 0.8) {\scriptsize 1};
    \node[vertex] (2) at (1, 0) {\scriptsize 2};
    \node[vertex] (3) at (1, -0.8) {\scriptsize 3};
    \node[vertex] (4) at (2, 0) {\scriptsize 4};
    \node[vertex] (5) at (3, 0.8) {\scriptsize 5};
    \node[vertex] (6) at (3, -0.8) {\scriptsize 6};

    \draw (0) -- (1);
    \draw (0) -- (2);
    \draw (0) -- (3);
    \draw (1) -- (4);
    \draw (2) -- (4);
    \draw (3) -- (4);
    \draw (4) -- (5);
    \draw (4) -- (6);
    \draw (5) -- (6);
\end{tikzpicture}
        \caption{}
        \label{fig:combinatorial-bound-better}
    \end{subfigure}
    \hfill
    \begin{subfigure}{0.49\columnwidth}
        \centering
        \begin{tikzpicture}
    \node[vertex] (0) at (0, 0) {\scriptsize 0};
    \node[vertex] (1) at (0.7, 0.8) {\scriptsize 1};
    \node[vertex] (2) at (0.7, -0.8) {\scriptsize 2};
    \node[vertex] (3) at (1.6, 0.5) {\scriptsize 3};
    \node[vertex] (4) at (1.6, -0.5) {\scriptsize 4};
    \node[vertex] (5) at (2.5, 0) {\scriptsize 5};
    \draw (0) -- (1);
    \draw (0) -- (2);
    \draw (1) -- (3);
    \draw (2) -- (4);
    \draw (3) -- (4);
    \draw (3) -- (5);
    \draw (4) -- (5);
\end{tikzpicture}
        \caption{}
        \label{fig:lp-bound-better}
    \end{subfigure}
    \caption{For the graph depicted in \subref{fig:combinatorial-bound-better}, the combinatorial bound ($3$) is stronger than the LP bound ($\tfrac{7}{4}$).
    For the graph depicted in \subref{fig:lp-bound-better}, the LP bound ($\tfrac{5}{4}$) is stronger than the combinatorial bound ($1$). 
    For details, see \Cref{example:combinatorial-bound-better,example:lp-bound-better}.}
    \label{fig:bound-comarison}
\end{figure}

\begin{example}\label{example:combinatorial-bound-better}
    Consider the graph $G=(V,E)$ depicted in \Cref{fig:combinatorial-bound-better}.
    Suppose there exists a partition $\Pi$ with maximum disagreement $\varphi(\Pi) \leq 2$. 
    By \Cref{lem:main-lemma} \ref{item:cluster-lower-bound}, all pairs of nodes whose neighborhoods share at least five nodes must be in the same cluster in $\Pi$. 
    In the given graph, the neighborhoods of all pairs of nodes share at most 3 nodes, i.e. $\Pi_2 = \{\{v\} \mid v \in V\}$ consists of singleton clusters. 
    By \Cref{lem:main-lemma} \ref{item:cluster-upper-bound}, all pairs of nodes whose neighborhoods have a symmetric difference containing five or more nodes must be in distinct clusters of $\Pi$.
    For the given graph, this implies that $4$ must be in a cluster different from that of $1$, $2$, and $3$, i.e $U^2_{\{4\}} = \{0, 4, 5, 6\}$.
    We can bound the disagreement of $4$ with respect to partition $\Pi$ by 
    \begin{align*}
        |N_4 \triangle [4]_\Pi| &= |N_4 \setminus [4]_\Pi| + |[4]_\Pi \setminus N_4| 
        \\
        &\geq |N_4 \setminus \{0, 4, 5,6\} | + |\{4\} \setminus N_4| = 3
    \end{align*}
    which is strictly greater than $2$, in contradiction to the assumption. 
    Therefore, every partition has maximum disagreement at least $3$.
    In fact, this is the optimal value assumed e.g.~by the partition $\Pi = \{\{0, 1\}, \{2\}, \{3\}, \{4, 5, 6\}\}$.

    In contrast, the LP bound is $\tfrac{7}{4}$, which is strictly less than the combinatorial bound of $3$.
    The LP bound is assumed by the solution $x_{\{4,5\}} = 0$, $x_{\{0,1\}} = x_{\{0,2\}} = x_{\{0,3\}} = x_{\{1,4\}} = x_{\{1,5\}} = x_{\{2,4\}} = x_{\{2,5\}} = x_{\{3,4\}} = x_{\{3,5\}} = \tfrac{1}{2}$, $x_{\{1,6\}} = x_{\{2,6\}} = x_{\{3,6\}} = \tfrac{3}{4}$, and $x_e = 1$ for all other edges $e$.
    The maximum disagreement of $\tfrac{7}{4}$ is obtained at node $4$.    
\end{example}

\begin{example}\label{example:lp-bound-better}
    Consider the graph $G=(V,E)$ depicted in \Cref{fig:lp-bound-better}.
    Suppose there exists a partition $\Pi$ with disagreement $\varphi(\Pi) = 0$. 
    Then, by \Cref{lem:main-lemma} \ref{item:cluster-lower-bound}, all pairs of nodes whose neighborhoods intersect in at least one node must be in the same cluster in $\Pi$.
    For the given graph, this implies that all nodes are in the same cluster, i.e.~$\Pi_0 = \{V\}$.
    However, this partition has disagreement $3$, in contrast to the assumption.
    Thus, there cannot exist a partition of with disagreement $0$.
    
    Next, suppose there exists a partition $\Pi$ with disagreement $\varphi(\Pi) = 1$.
    By \Cref{lem:main-lemma} \ref{item:cluster-lower-bound}, all pairs of nodes whose neighborhoods share at least three nodes must be in the same cluster in $\Pi$.
    This implies that the nodes $3$, $4$, and $5$ are in the same cluster, i.e. $\Pi_1 = \{\{0\}, \{1\}, \{2\}, \{3, 4, 5\}\}$.
    \Cref{lem:main-lemma} \ref{item:cluster-upper-bound} implies that all pairs of nodes whose neighborhoods have a symmetric difference containing more than three nodes must be in distinct clusters of $\Pi$.
    For the given graph, we get $U^1_{\{0\}} = \{0, 1, 2\}$, $U^1_{\{1\}} = \{0, 1\}$, $U^1_{\{2\}} = \{0, 2\}$ and $U^1_{\{3, 4, 5\}} = \{3, 4, 5\}$.
    For all nodes, the bound \eqref{eq:disagreement-bound} is less than or equal to $1$.
    Therefore, the combinatorial lower bound of the given graph is $1$.

    In contrast, the LP bound is $\tfrac{5}{4}$ which is strictly greater than the combinatorial lower bound of $1$.
    The LP bound is assumed by the solution $x_{\{0,1\}} = x_{\{0,2\}} = \tfrac{1}{2}$, $x_{\{1,3\}} = x_{\{2,4\}} = \tfrac{3}{4}$, $x_{\{3,4\}} = x_{\{3,5\}} = x_{\{4,5\}} = \tfrac{1}{4}$, and $x_e = 1$ for all other edges.
    The maximum disagreement of $\tfrac{5}{4}$ is obtained at nodes $3$ and $4$.
\end{example}

\subsection{Non-complete and weighted graphs}

The combinatorial lower bound can be adapted to the case of non-complete graphs.
However, \Cref{lem:main-lemma} \ref{item:cluster-upper-bound} needs to be adjusted:
If the number of nodes that are in the negative neighborhood of $u$ and in the positive neighborhood of $v$, or vice versa, is greater than $2 \varphi(\Pi)$, then $u$ and $v$ must be in distinct clusters.
This leads to a different definition of $U^d$ in \Cref{def:lower-upper-cluster-bound} which in turn leads to a different definition of the combinatorial lower bound in \Cref{thm:clb}.
For this adapted version, \Cref{lem:unique-cluster}, \Cref{cor:unique-cluster}, and \Cref{thm:four-approximation} no longer hold.
Whether this bounding technique for non-complete graphs can be used to derive an approximation algorithm is an open problem.

The bounding technique can even be adapted to the weighted variant of the min max correlation clustering problem.
Again, the main difference lies in \Cref{lem:main-lemma}.
For two nodes $u$, $v$ in distinct clusters of a partition, each node $w$ in the positive neighborhood of both $u$ and $v$ contributes to either the disagreement of $u$ or $v$.
Therefore, $w$ contributes a disagreement of at least $\min(\theta_{uw}, \theta_{vw})$ where $\theta_e$ is the weight of edge $e$.
As a consequence, the maximum disagreement of $u$ and $v$ is at least $\tfrac{1}{2}\sum_{w \in N_u \cap N_v} \min(\theta_{uw}, \theta_{vw})$.
This yields an analogue to \Cref{lem:main-lemma} \ref{item:cluster-lower-bound} for the weighted case:
If $\sum_{w \in N_u \cap N_v} \min(\theta_{uw}, \theta_{vw}) > 2 \varphi(\Pi)$, then $[u]_\Pi = [v]_\Pi$.
Similarly, an analogue to \Cref{lem:main-lemma} \ref{item:cluster-upper-bound} can be derived.

\section{GREEDY JOINING ALGORITHM}\label{sec:algorithm}

The $4$-approximation due to \Cref{thm:four-approximation} only adds a cluster if it decreases the maximum disagreement of the nodes contained in it by at least a factor of four.
Nodes with smaller degree may remain in singleton clusters.
For example, all graphs from practical applications we consider in \Cref{sec:experiments} are such that the largest node degree is less than four times the combinatorial lower bound.
For these graphs, the partition into singleton clusters is a $4$-approximation.
In practical applications like these, it is desirable not only to find a feasible solution that satisfies the approximation guarantee but also to further improve this feasible solution by local search.

In this section, we introduce an algorithm that seeks to iteratively improve a given partition by greedily joining clusters.
While we do not establish any improved approximation guarantee, the feasible solutions found by this algorithm improve upon the empirical state of the art for the applications we consider in \Cref{sec:experiments}.

The algorithm is remarkable simple:
In each iteration, a node $w$ with largest disagreement with respect to the current partition $\Pi$ is chosen.
Every neighbor $v$ of $w$ that is not in the same cluster as $w$ in $\Pi$ contributes to the disagreement for $w$.
The disagreement of $w$ can potentially be decreased by joining the cluster of $w$ with the cluster of $v$.
If no such neighbor $v$ exists, the disagreement of $w$ cannot be improved by joining two clusters, and the algorithm terminates.
Otherwise, we choose one such $v$ that fits well to $w$ according to \Cref{lem:main-lemma}, i.e.~such that the intersection of neighborhoods $|N_w \cap N_v|$ is large and the symmetric difference of $|N_w \triangle N_v|$ neighborhoods is small.
This algorithm is detailed in \Cref{algo:greedy}. 

\begin{algorithm}\small
    \DontPrintSemicolon
    \SetAlgoNoEnd
    \LinesNumbered
    \KwData{Graph $G=(V,E)$}
    \KwResult{Partition $\Pi$ of $V$}
    Let $\Pi$ be the partition computed by the $4$-approximation algorithm (\Cref{thm:four-approximation}) \;
    \While{True}{\label{line:while}
        $w \in \argmax_{v \in V} |N_v \triangle [v]_\Pi|$ \label{line:select-worst} \; 
        $made\_join = False$ \;
        \For{$v \in N_w$ sorted by $|N_w \cap N_v| - |N_w \triangle N_v|$ in descending order}{\label{line:sort-neighborhood}
            \If{$[v]_\Pi = [w]_\Pi$}{
                \Continue \tcp{$v$ is already in the cluster of $w$.}
            }
            $C = [v]_\Pi \cup [w]_\Pi$ \tcp{join the clusters of $v$ and $w$}
            $d = \max_{u \in C} |N_u \triangle C|$ \label{line:compute-disagreement} \tcp{maximal disagreement of all nodes in $C$}
            \If{$d > |N_v \triangle [v]_\Pi|$}{\label{line:discard-join}
                \Continue \tcp{the maximal disagreement in $C$ is greater than the current maximal disagreement}
            }
            $\Pi := \Pi \setminus \{[v]_\Pi, [w]_\Pi\} \cup \{C\}$ \tcp{replace the clusters of $v$ and $w$ by $C$.}
            $made\_join = True$ \;
            \textbf{break} \;
        }
        \If{not $made\_join$}{
            \KwRet{$\Pi$}
        }
    }
    \caption{Greedy Joining}
    \label{algo:greedy}
\end{algorithm}

In several places of \Cref{algo:greedy}, ties can occur:
In \Cref{line:select-worst}, multiple nodes with largest agreement can exist.
In \Cref{line:sort-neighborhood}, multiple neighbors $v$ of $w$ can have the same value $|N_w \cap N_v| - |N_w \triangle N_v|$.
We break these ties by considering node degree, as the secondary ordering criterion, and node index, as the tertiary ordering criterion.
In \Cref{line:select-worst}, we select among all nodes with largest disagreement a node with largest/smallest degree and refer to these two options as Design Choice 1 (DC1).
In \Cref{line:sort-neighborhood}, we sort neighbors $v$ of $w$ with the same value $|N_w \cap N_v| - |N_w \triangle N_v|$ by increasing/decreasing degree (DC2).
We investigate two more choices in the design of \Cref{algo:greedy}:
Instead of sorting the neighbors of $w$ by $|N_w \cap N_v| - |N_w \triangle N_v|$, we consider sorting these only by $|N_w \cap N_v|$ or only by $-|N_w \triangle N_v|$ (DC3).
Finally, we consider making \Cref{line:discard-join} stricter by discarding a join if there exists $u \in C$ such that the disagreement of $u$ with $\Pi$ is strictly less than that of $w$, but the disagreement of $u$ with $C$ is equal to the disagreement of $w$ with $\Pi$ (DC4).

We denote by $\mathcal{A}$ the variant of \Cref{algo:greedy} where ties are resolved in favor of largest degree, where neighborhoods are sorted by $|N_w \cap N_v| - |N_w \triangle N_v|$ and with the stricter version of \Cref{line:discard-join}.
We denote by $\mathcal{A}^*$ the algorithm that executes \Cref{algo:greedy} for all 24 combinations of DC1-4 and outputs the best solution.

\begin{lemma}\label{lem:greedy-complexity}
    \Cref{algo:greedy} has worst case time complexity of $\mathcal{O}(n^2 \delta^2)$.
\end{lemma}

\begin{proof}
    The algorithm terminates after $\mathcal{O}(n)$ joins.
    Each join can be computed in time $\mathcal{O}(n\delta^2)$:
    There are $|N_w| \leq \delta$ candidate clusters that can be joined with the cluster of $w$ (\Cref{line:sort-neighborhood}).
    For each candidate cluster, the maximum disagreement with respect to the joint cluster (\Cref{line:compute-disagreement}) is computed in time $\mathcal{O}(n\delta)$.
\end{proof}

Despite its high worst case time complexity, \Cref{algo:greedy} is efficient in practice (see \Cref{sec:experiments}) for the following reason:
In real world instances there are typically only a few nodes with large degree and thus large disagreement. 
Therefore, \Cref{algo:greedy} only operates on a few nodes and their neighborhoods.
In particular, it is not necessary to compute the size of the intersections of the neighborhoods of all pairs of nodes.
Note, however, that these computations are required in order to compute the combinatorial lower bound as well as in the algorithm by \citet{davies2023fast}.

\section{EXPERIMENTS}\label{sec:experiments}

We follow \citet{davies2023fast} and evaluate the combinatorial lower bound and the greedy joining algorithm on social network graphs as well as synthetic graphs. 
We compare the results to those of \citet{davies2023fast}.
All experiments were performed on a Lenovo X1 Carbon laptop equipped with an Intel Core i7-10510U CPU @ 1.80GHz and 16 GB LPDDR3 RAM.

\subsection{Datasets}

\paragraph*{Social network graphs}

The \emph{ego-Facebook} dataset \citep{mcauley2012learning} contains ten graphs that represent circles of friends from the social network Facebook.
The graphs contain 52 to 1,034 nodes and 146 to 30,025 edges.
The \emph{feather-lastfm-social} dataset \citep{feather} is a social network of 7,624 LastFM users with 27,806 edges.
The \emph{ca-HepPh} and \emph{ca-HepTh} datasets \citep{leskovec2007graph} are collaboration networks of authors in the field of high energy physics. 
They contain 12,008 and 9,977 nodes and 118,489 and 25,998 edges, respectively.
Lastly, the \emph{com-Youtube} dataset is a social network of 1,134,890 Youtube users with 2,987,624 edges.
These datasets are available online\footnote{\url{https://snap.stanford.edu/data/}}.

\paragraph*{Synthetic graphs}

In order to compare algorithms for the min max correlation clustering problem also in a more controlled setting, \Citet{davies2023fast} synthesize instances of the problem as follows:
Beginning with 10 distinct cliques containing 10 nodes each, a given number of pairs of nodes are selected at random. A selected pair is added to the edge set in case the nodes are not in a clique, and removed from the edge set, otherwise. I.e., edge are flipped from attractive to repulsive, or vice versa.
The number of disagreements increases with the number of flips.
For each number $f \in \{0,50,100,\dots,1000\}$ of flips, we construct 10 instances randomly, in this way.

\paragraph*{Results}

In \Cref{tab:ego-Facebook,tab:other}, we report empirical results for the social network graphs. 
In Columns 1-4, we report the name of the dataset, the number of nodes, number of edges and largest degree of the graph.
In the columns $\clb$ and LP, we report the combinatorial lower bound and the LP bound.
In the columns $\mathcal{A}$ and $\mathcal{A}^*$, we report the maximum disagreements of the partitions found by algorithms $\mathcal{A}$ and $\mathcal{A}^*$, as described in \Cref{sec:algorithm}.
In the columns $\dmn$ and $\kmz$, we report the maximum disagreements of the partitions found by the algorithm of \citet{davies2023fast} and \citet{kalhan2019correlation}.
Note that the rounding algorithm utilized by the $\dmn$ and $\kmz$ algorithm is parameterized by two parameters. We denote by $\dmn$ and $\kmz$ the results that are obtained by using the same parameters for all problem instances and denote by $\dmn^*$ and $\kmz^*$ the results that are obtained by searching for the best parameters for each problem instances individually (for details, see \citet{davies2023fast}).
In the last columns, we report the times for computing the combinatorial lower bound, for computing the LP bound, and for running the algorithms $\mathcal{A}$, $\mathcal{A}^*$, and $\dmn$.
As our algorithms are implemented in c++ and the $\dmn$ algorithm of \citet{davies2023fast} is originally implemented in python, we contribute, in addition, a performance optimized \textsc{c}++ implementation of $\dmn$ that exploits sparsity of the input graph.
We report the runtime of the original python implementation ($t_{\dmn}$), as well as that of our \textsc{c}++ implementation ($t_{\dmn}^{++}$).
We do not report the runtime of the $\kmz$ algorithm explicitly as it is dominated by solving the LP ($t_\text{LP}$).
The LP is solved with \Citet{gurobi}, using the barrier method.

\begin{table*}
    \centering
    \setlength\tabcolsep{3pt}
    \small
    \begin{tabular}{r r r r r r r r r r r r r r r r r r}
        \toprule
        ID & $|V|$ & $|E|$ & $\delta$ & $\clb$ & LP & $\mathcal{A}$ & $\mathcal{A}^*$ & $\dmn$ & $\dmn^*$ & $\kmz$ & $\kmz^*$ & $t_{\clb}$ & $t_\text{LP}$ &$t_\mathcal{A}$ & $t_{\mathcal{A}^*}$ & $t_{\dmn}$ & $t_{\dmn}^{++}$ \\
        \midrule
           0 &  333 &  2519 &  77 &  32 &     - &  46 &  44 &  49 &  49 &  - &  - &   5.32 &              - &  0.36 &  10.4 &   159 &   6.58 \\
         107 & 1034 & 26749 & 253 &  95 &     - & 123 & 122 & 152 & 134 &  - &  - & 112.38 &              - &  3.30 &  84.1 & 1,294 & 230.05 \\
         348 &  224 &  3192 &  99 &  39 & 39.13 &  61 &  50 &  72 &  71 & 89 & 69 &   2.12 & $1.5\cdot10^6$ &  0.53 &  12.4 &   117 &  13.04 \\
         414 &  150 &  1693 &  57 &  18 & 19.66 &  27 &  27 &  34 &  31 & 38 & 28 &   2.06 & $2.0\cdot10^5$ &  0.54 &  14.5 &    72 &   3.27 \\
         686 &  168 &  1656 &  77 &  31 & 30.48 &  45 &  43 &  47 &  43 & 69 & 47 &   1.06 & $4.1\cdot10^5$ &  0.21 &   6.7 &    84 &   4.17 \\
         698 &   61 &   270 &  29 &  11 & 10.64 &  16 &  16 &  20 &  18 & 18 & 17 &   0.20 & $1.6\cdot10^3$ &  0.06 &   1.5 &    13 &   0.27 \\
        1684 &  786 & 14024 & 136 &  52 &     - &  80 &  78 &  93 &  93 &  - &  - &  65.76 &              - &  2.11 &  50.0 &   814 &  67.03 \\
        1912 &  747 & 30025 & 293 & 118 &     - & 166 & 163 & 220 & 187 &  - &  - &  41.66 &              - &  4.63 & 116.4 &   877 & 272.84 \\
        3437 &  534 &  4813 & 107 &  49 &     - &  58 &  57 & 107 &  77 &  - &  - &   7.79 &              - &  0.89 &  14.9 &   405 &  16.24 \\
        3980 &   52 &   146 &  18 &   8 &  7.34 &  11 &  11 &  12 &  12 & 13 & 13 &   0.10 & $6.7\cdot10^2$ &  0.07 &   1.2 &     7 &   0.20 \\
        \bottomrule
        \end{tabular}
    \caption{This table summarizes the results for the \emph{ego-Facebook} graphs.
    The runtime is reported in milliseconds.
    The LP bound and $\kmz$ objective is only reported for the five smallest instances as solving the LP for the graph 348 already takes approximately 25 minutes.}
    \label{tab:ego-Facebook}
\end{table*}

\begin{table*}
    \centering
    \setlength\tabcolsep{3pt}
    \small
    \begin{tabular}{r r r r r r r r r r r r r r}
        \toprule
        Name & $|V|$ & $|E|$ & $\delta$ & $\clb$ & $\mathcal{A}$ & $\mathcal{A}^*$ & $\dmn$ & $\dmn^*$ & $t_{\clb}$ &$t_\mathcal{A}$ & $t_\text{$\dmn$}$ & $t_{\dmn}^{++}$ \\
        \midrule
        \emph{lastfm}  &     7,624 &    27,806 &   216  & 106 &    116 &    116 & 216 & 160 & 1,724 &           8.74 & $6.4\cdot10^4$ &   281 \\
        \emph{HepPh}   &    12,006 &   118,489 &   491  & 208 &    251 &    250 & 333 & 267 & 8,023 &          27.91 & $1.2\cdot10^5$ & 2,410 \\
        \emph{HepTh}   &     9,875 &    25,973 &    65  &  34 &     42 &     42 &  65 &  58 & 3,611 &           8.52 & $1.7\cdot10^5$ &   212 \\
        \emph{Youtube} & 1,134,890 & 2,987,624 & 28,754 &   - & 15,486 & 14,796 &   - &   - &   -   & $4.1\cdot10^4$ &              - &     - \\
        \bottomrule
        \end{tabular}
    \caption{This table summarizes the results for the \emph{feather-lastfm-social}, \emph{ca-HepPh}, \emph{ca-HepTh}, and \emph{com-Youtube} graphs.
    The runtime is reported in milliseconds.
    For the largest graph, the combinatorial lower bound and $\dmn$ are not computed as this would exceed time and memory constraints.}
    \label{tab:other}
\end{table*}

In \Cref{fig:synthetic}, we report bounds, maximum disagreements and runtimes for the synthetic instances of the min max correlation clustering problem.
In this figure, thick lines indicate the median across the 10 random instances, shaded areas indicate the second and third quartile, and dashed lines indicate the 0.1 and 0.9-quantile.

\begin{figure}
    \centering
    \includegraphics[width=\columnwidth]{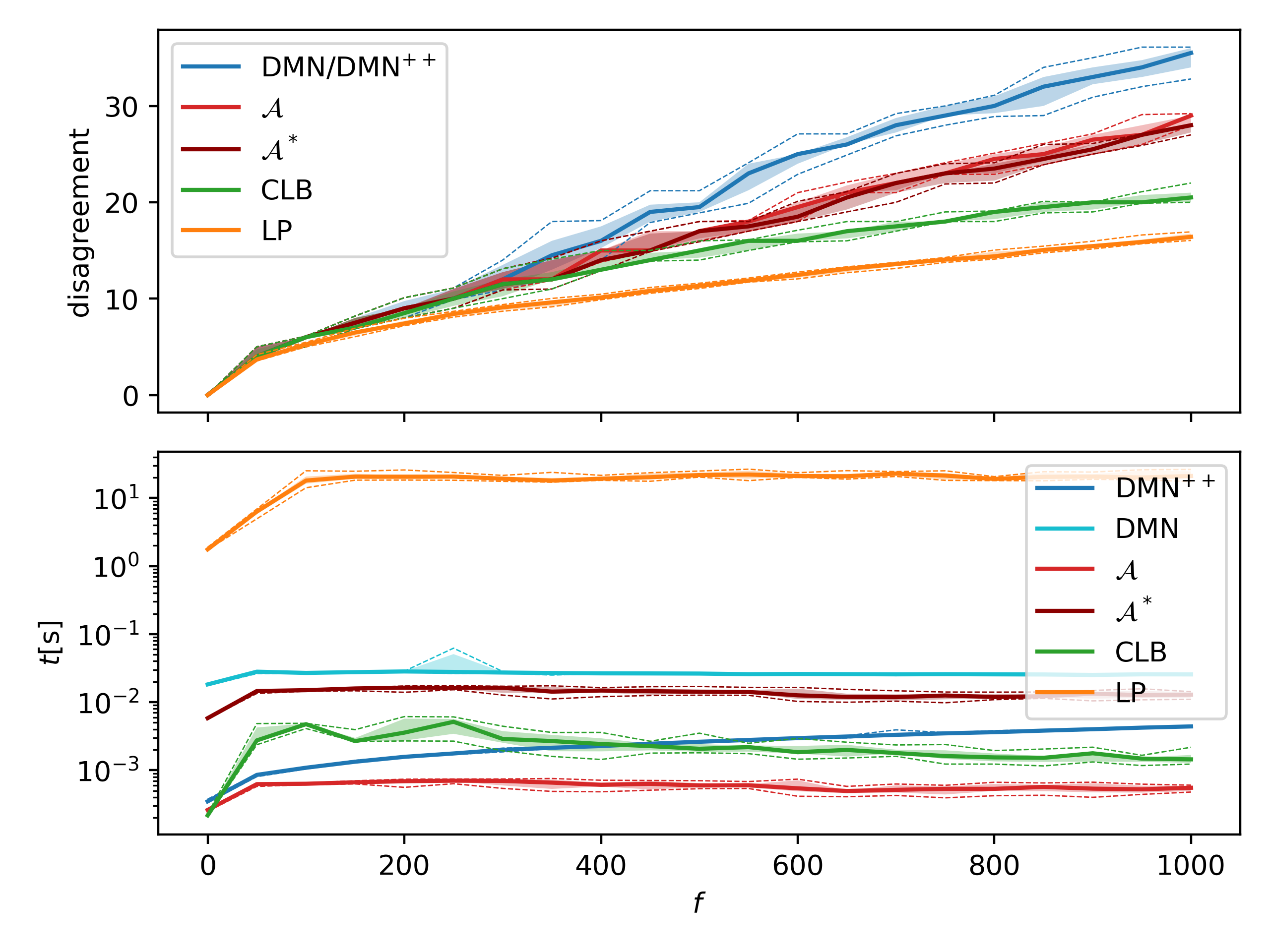}
    \caption{Depicted above are the maximum disagreements of the partitions computed by the $\dmn$ algorithm and our algorithms $\mathcal{A}$, $\mathcal{A}^*$ as well as the lower bounds according to the combinatorial lower bound ($\clb$) and the LP bound.
    Depicted below are the runtimes in seconds of all algorithms and both bounding techniques. $f$ is the number of random flips in a graph with 100 nodes (10 cliques of size 10).}
    \label{fig:synthetic}
\end{figure}

\paragraph*{Analysis}

As can be seen from \Cref{tab:ego-Facebook,tab:other}, the partitions computed by $\mathcal{A}$ have a lower maximum disagreement than those computed by $\dmn$ and $\kmz$, across all instances.
For the Facebook ego-network with ID 3437, the disagreement of the partition computed by $\mathcal{A}$ is $58$, compared to $107$ computed by $\dmn$.
The runtime of $\mathcal{A}$ is approximately one order of magnitude shorter than that of $\dmn$.
This is due to the fact that in $\dmn$, the intersections of neighborhoods of all pairs of nodes need to be computed, while in $\mathcal{A}$, these intersections only for the node with the largest disagreement and its neighbors need to be computed (\Cref{line:sort-neighborhood}).
The disagreement of partitions computed by $\mathcal{A}^*$ is often strictly less than that computed by $\mathcal{A}$.
The greatest relative improvement of $\mathcal{A}^*$ over $\mathcal{A}$ can be observed on the Facebook ego-network with ID 348.
Here, the disagreement of $50$ is achieved by the greedy joining algorithm in which neighborhoods are sorted by $-|N_v \triangle N_w|$.
Similarly, the disagreement of partitions computed by $\dmn^*$ and $\kmz^*$ are often strictly less than that computed by $\dmn$ and $\kmz$. 
The disagreements achieved by $\mathcal{A}^*$ are less than that of $\dmn^*$ and $\kmz^*$ on all instances except one where there is a tie.

The combinatorial lower bound ($\clb$) and the LP bound are similar. 
However, the combinatorial lower bound can be computed many orders of magnitude faster than the LP bound.
For the the Facebook ego-network with ID 348, solving the LP requires approximately 25 minutes while computing the combinatorial lower bound takes approximately 2 milliseconds.
The fact that the combinatorial lower bound can be computed so much more efficiently allows us to compute the first non-trivial lower bounds for larger instances. 
This includes the five large instances of the \emph{ego-Facebook} dataset for which we cannot report an LP bound in \Cref{tab:ego-Facebook}, and it includes the even larger graphs in \Cref{tab:other}.
Only for the \emph{com-Youtube} graph in \Cref{tab:other} with more than one million nodes have we found the computation the combinatorial lower bound to be impractical.
Across all other instances, the greatest relative gap between the maximum disagreement of the partition computed by $\mathcal{A}$ and the lower bound is 1.56 (Facebook ego-network with ID 348).

For synthetic graphs (\Cref{fig:synthetic}), the results are similar:
For small numbers of flips, both $\dmn$ and $\mathcal{A}$ compute the optimal partition, which can be seen from the fact that there is no gap between the optimal solution and the lower bound.
For larger numbers of flips, the gap between the lower bounds and the maximum disagreement of the computed partitions increases.
Notably, the gap of the partitions computed by $\mathcal{A}$ is approximately half of the gap of the partitions computed by $\dmn$.
The combinatorial lower bound ($\clb$) is slightly stronger than the LP bound.
The runtime of the greedy joining algorithm $\mathcal{A}$ is about one order of magnitude shorter than that of $\dmn$.
The difference in runtime is even greater between $\clb$ and the LP bound.
Similar to the social network graphs, the greatest relative gap we observe between $\clb$ and the maximum disagreement of the partition computed by $\mathcal{A}$ is approximately $1.5$.
In fact, we have not found instances where the relative gap between $\clb$ and the maximum disagreement of the partition computed by $\mathcal{A}$ is greater than $2$.

\section{CONCLUSION}\label{sec:conclusion}

We have introduced a combinatorial lower bounding technique for the min max correlation clustering problem for complete graphs.
There are instances where this bound is stronger than the canonical LP~bound, and vice versa.
This motivates future work to combine these bounds.
To this end, the constraints from \Cref{lem:main-lemma} can be expressed in the form of quadratic inequalities.
However, we have not observed improvements over the LP~bound when adding linear relaxations of these quadratic inequalities to the LP (results not shown).
From the combinatorial lower bound, we have derived a $4$-approximation that we have extended by a greedy local search heuristic.
On all instances we have considered in the experiments for this article, the greedy joining algorithm yields a $2$-approximation.
Whether the greedy joining algorithm is indeed a $2$-approximation algorithm is an open problem.
We have discussed briefly generalizations of the combinatorial lower bound to non-complete and weighted graphs.
Whether approximation guarantees can be derived for these cases, or other objectives than the min max objective discussed in this article, remains open.

\bibliography{references}
\bibliographystyle{plainnat}

\end{document}